\documentclass[11pt]{article}

\usepackage[letterpaper,margin=1in,bottom=1in]{geometry}

\usepackage{amsmath,amssymb,amsthm}
\usepackage{microtype}
\usepackage{xcolor}
\usepackage{array}
\usepackage{float}
\usepackage{algorithm}
\usepackage{algpseudocode}
\usepackage{authblk}
\algrenewcommand\algorithmicrequire{\textbf{Input:}}
\algrenewcommand\algorithmicensure{\textbf{Output:}}
\usepackage{tikz}
\usetikzlibrary{arrows.meta,calc,positioning}
\usepackage{quantikz}
\usepackage{hyperref}
\usepackage[capitalise]{cleveref}
\usepackage{physics}

\floatstyle{ruled}
\restylefloat{algorithm}

\hypersetup{
  colorlinks=true,
  linkcolor=blue!55!black,
  urlcolor=blue!55!black,
  citecolor=blue!55!black,
  pdftitle={Gate-Efficient Implementation of Query-Optimal
  Time-Dependent Hamiltonian Simulation}
}
\allowdisplaybreaks

\newtheorem{theorem}{Theorem}
\newtheorem{lemma}[theorem]{Lemma}
\newtheorem{proposition}[theorem]{Proposition}
\theoremstyle{definition}
\newtheorem{definition}{Definition}

\theoremstyle{remark}
\newtheorem*{remark}{Remark}
\theoremstyle{plain}

\newcommand{\Sys}{\mathcal H_{\mathsf S}}
\newcommand{\Anc}{\mathcal H_{\mathsf A}}
\newcommand{\Time}{\mathcal H_{\mathsf T}}
\newcommand{\Priv}{\mathcal H_{\mathrm{priv}}}
\newcommand{\HAMT}{\mathrm{HAM\mbox{-}T}}
\newcommand{\eps}{\varepsilon}
\renewcommand{\norm}[1]{\left\lVert #1\right\rVert}
\renewcommand{\ket}[1]{\lvert #1\rangle}
\renewcommand{\bra}[1]{\langle #1\rvert}
\renewcommand{\ketbra}[2]{\lvert #1\rangle\!\langle #2\rvert}
\newcommand{\bigO}{\mathcal O}

\DeclareRobustCommand{\affilstyle}[1]{{\small\itshape #1}}

\title{Gate-Efficient Implementation of the Query-Optimal\\
Time-Dependent Hamiltonian Simulation}

\author[1]{Boyang Chen\thanks{\texttt{by-chen24@mails.tsinghua.edu.cn}}}
\author[2,3]{Minbo Gao\thanks{\texttt{gmb17@tsinghua.org.cn}}}
\author[1]{Zhengfeng Ji\thanks{\texttt{jizhengfeng@tsinghua.edu.cn}}}
\author[4,5]{Tongyang Li\thanks{\texttt{tongyangli@pku.edu.cn}}}
\author[4,5]{Xinzhao Wang\thanks{\texttt{xinzhaowang3@gmail.com}}}
\author[4,5]{Shuo Zhou\thanks{\texttt{antientropy@pku.edu.cn}}}

\affil[1]{\affilstyle{Department of Computer Science and Technology,
Tsinghua University, Beijing, China}}

\affil[2]{\affilstyle{Institute of Software,
Chinese Academy of Sciences, Beijing, China}}

\affil[3]{\affilstyle{University of Chinese Academy of Sciences,
Beijing, China}}

\affil[4]{\affilstyle{Center on Frontiers of Computing Studies,
Peking University, Beijing, China}}

\affil[5]{\affilstyle{School of Computer Science,
Peking University, Beijing, China}}
\date{}

\begin{document}
\maketitle
\vspace{-2.2em}

\begin{abstract}
The query-optimal algorithm of \cite{CGWZ26} for general time-dependent
Hamiltonian simulation uses
\begin{equation*}
  q = \bigO\! \left( \alpha T +
      \frac{\log(1/\eps)}{\log\! \left(e + \log(1/\eps)/(\alpha T) \right)}
      \right)
\end{equation*}
queries to $\HAMT$ within $\eps$ error for a Lipschitz-continuous time-dependent
Hamiltonian $H(t)$ on $[0,T]$ satisfying $\norm{H(t)}\leq\alpha$.
However, its direct circuit implementation incurs a substantially larger gate
overhead.
In this note, we give an implementation of the same algorithm that retains its
optimal query complexity and uses
\begin{equation*}
  \bigO\! \left[ q \left( a + \log\! \left(1 + \frac{T(\alpha + \beta T)}{\eps}
      \right) \right) \right]
\end{equation*}
one- and two-qubit gates, where $a$ is the number of block-encoding ancilla
qubits and $\beta$ is the Lipschitz constant of $H$.
The main ingredient is an exact dyadic factorization of the ordered update
product in the underlying one-query transducer.
\end{abstract}

\newpage
\section{Introduction}

Hamiltonian simulation has been a central topic in quantum algorithm design
since the early development of quantum computation.
A long line of work has progressively improved its dependence on the evolution
time and the target precision.
Early algorithms were based on product formulas and sparse-Hamiltonian
simulation techniques~\cite{AT03,BACS07}, followed by approaches using linear
combinations of unitaries and truncated series~\cite{CW12,BCK15}.
A major conceptual advance came with quantum signal processing and
qubitization~\cite{LC17,LC19}, which, together with the more general framework
of quantum singular value transformation~\cite{GSLW19Full}, led to an
essentially complete understanding of the query complexity of time-independent
Hamiltonian simulation.
For a Hamiltonian $H$ satisfying $\norm{H}\leq \alpha$, the optimal number of
queries required to simulate $e^{-iHT}$ to precision $\eps$ is
\begin{equation*}
  \Theta\! \left(\alpha T +
      \frac{\log(1/\eps)}{\log\! \left(e + \log(1/\eps)/(\alpha T) \right)}
      \right).
\end{equation*}
Thus, even in the time-independent setting, reaching the optimal joint
dependence on time and precision required a line of increasingly refined
algorithmic ideas.

The corresponding problem for time-dependent Hamiltonians has proved more
subtle.
Given $H(t)$ on $[0,T]$, the target evolution is the time-ordered propagator
\begin{equation*}
  U_H(T) = \mathcal{T} \exp\! \left(-i \int_0^T H(t) \dd{t} \right),
\end{equation*}
and the noncommutativity of $H(t)$ at different times makes the time ordering an
essential part of the problem.
For this problem, product-formula methods have a long history in time-dependent
Hamiltonian simulation~\cite{HR90,WBHS10,PQSV11}, while later algorithms based
on truncated Dyson series~\cite{LW18,KSB19,BCS+20}, enlarged clock or Floquet
spaces~\cite{MizutaFujii23,Mizuta23,WatkinsEtAl24,LiWang25,ZLA26}, and Magnus
expansions~\cite{AnFangLin22,FangLiuSarkar25,FangLiuZhu25} substantially
improved the simulation complexity in different settings.
Despite this progress, for a general time-dependent Hamiltonian, none of these
approaches attained the optimal query complexity known in the time-independent
setting.
Thus, whether time dependence intrinsically incurs an additional query overhead
remained open long after the time-independent problem had been settled.

A query-optimal algorithm for time-dependent Hamiltonian simulation was
recently given in~\cite{CGWZ26}.
For a time-dependent Hamiltonian satisfying $\norm{H(t)}\leq\alpha$ and the
$\beta$-Lipschitz condition, the algorithm approximates $U_H(T)$ to
error $\eps$ using
\begin{equation*}
  q = \bigO\! \left(\alpha T +
      \frac{\log(1/\eps)}{\log\! \left(e + \log(1/\eps)/(\alpha T) \right)}
      \right)
\end{equation*}
$\HAMT$ queries~\cite{CGWZ26}.
This matches the known lower bound for time-independent
Hamiltonians~\cite{LC17,GSLW19Full}, showing that time dependence incurs no
asymptotic overhead in query complexity.

Query optimality alone does not guarantee a comparably efficient circuit
implementation.
While the algorithm in~\cite{CGWZ26} uses only $\bigO(q)$ queries, its direct
implementation requires
\begin{equation*}
  \widetilde{\bigO}\! \left( q(a + 1) \left( q + \frac{\beta T^2}{\eps} +
      \frac{(\alpha T)^{3/2}}{\sqrt{\eps}} \right) \right)
\end{equation*}
one- and two-qubit gates, where $a$ is the number of block-encoding ancilla
qubits.
This raises the central question of the present work:
\begin{quote}
\emph{Does the query-optimal algorithm admit an efficient gate implementation?}
\end{quote}

In this note, we answer this question affirmatively.
More concretely, we show that
\begin{theorem}
\label{thm:informal-main-theorem}
The algorithm of~\cite{CGWZ26} can be implemented with
\begin{equation}
  \bigO\! \left[ q \left( a + \log\! \left(1 + \frac{T(\alpha + \beta T)}{\eps}
      \right) \right) \right]
  \label{eq:gate-complexity}
\end{equation}
one- and two-qubit gates, and with
\begin{equation*}
  \bigO\! \left( a + \log\! \left(1 + \frac{T(\alpha + \beta T)}{\eps} \right)
      \right)
\end{equation*}
auxiliary qubits.
\end{theorem}

The algorithm of~\cite{CGWZ26} is based on a one-query transducer containing the
ordered product $S^\circ=R_{J-1}\cdots R_0$, with one update for each of the $J$
sampled times.
For $J=2^m$, we give an exact dyadic factorization of $S^\circ$ into two basis
changes on the $m$-qubit time register, a rotation, and a phase gate.
Each basis change contains a single one-qubit gate for each time qubit,
controlled on the lower time qubits being zero.
These conditions can be computed successively, so a controlled implementation of
$S^\circ$ uses $\bigO(a+\log J)$ one- and two-qubit gates.
Combining this factorization with the implementations of $\operatorname{PREP}$
and $\operatorname{SELECT}$ gives the stated gate and auxiliary-qubit bounds.

The $\HAMT_J$ model imposes an auxiliary-space lower bound.
In the parameter regime of \cref{prop:auxiliary-lower-bound}, any algorithm in
the $\HAMT_J$ model with $t_j=jT/J$ requires
\begin{equation*}
  \Omega\! \left( a + \log\! \left(1 + \frac{\beta T^2}{\eps} \right) \right)
\end{equation*}
auxiliary qubits.  Hence the dependence on $a$ and
$\log(1+\beta T^2/\eps)$ in \cref{thm:informal-main-theorem} is necessary in
this model.

Under a stronger input model, the same circuit implementation also achieves
$L^1$-norm scaling by time reparameterization~\cite{BCS+20}.
In \cref{sec:l1-scaling}, we use a Lipschitz function $h(t)\geq\norm{H(t)}$ and
assume coherent oracle access to the resulting reparameterized Hamiltonian.
The query complexity then depends on
\begin{equation*}
  \Lambda_h := \int_0^T h(t) \dd{t},
\end{equation*}
which upper-bounds $\int_0^T\norm{H(t)}\dd t$.
Choosing $h(t)=\norm{H(t)}$ recovers scaling with the exact integrated norm.

\subsection{Open Questions}
Our note leaves the following questions for future investigation.

\paragraph{Optimality of the gate complexity}
For time-independent Hamiltonians, quantum signal processing gives query-optimal
simulation using $\bigO(q(a+1))$ one- and two-qubit gates.
For general time-dependent Hamiltonians, \Cref{thm:informal-main-theorem} gives
the bound \eqref{eq:gate-complexity}.
Is the additional logarithmic dependence in this bound necessary, or can it be
removed?

\paragraph{Reducing the number of ancilla qubits}
Unlike query-optimal time-independent simulation based on quantum signal
processing, our algorithm uses a transducer and an LCU construction and requires
logarithmically many ancilla qubits in addition to the block-encoding register.
\Cref{prop:auxiliary-lower-bound} shows that this logarithmic dependence is
unavoidable for general Lipschitz-continuous Hamiltonians in the $\HAMT_J$
model: one must have $J=\Omega(\beta T^2/\eps)$, and the time register
$\mathsf T$ alone contains $\log_2J$ qubits to label $H_0,\ldots,H_{J-1}$.
This leaves open whether additional structure in specific Hamiltonian families
(for example, linear time dependence in an adiabatic interpolation) can be
exploited to reduce the ancilla count.

\section{Preliminaries and review of the query-optimal algorithm}
\label{sec:preliminaries}

We review the query-optimal algorithm of \cite{CGWZ26}.
The algorithm first approximates the target evolution by a Cayley product $U_C$.
A one-query transducer then implements $U_C$ when supplied with a catalyst.
Applying finite reuse to this transducer gives block encodings of operators
$P_N$ that approximate $U_C$, without requiring the catalyst as input.
Combining these block encodings by the linear combination of unitaries (LCU)
technique reduces the reuse error, and one step of oblivious amplitude
amplification gives the final simulation circuit.

\subsection{Hamiltonian and \texorpdfstring{$\HAMT$}{HAM-T} access}
Let $\mathsf S$ be the system register with Hilbert space $\Sys$.
The target evolution $U_H(t)$ is defined by
\begin{equation*}
  i \dv{}{t} U_H(t) = H(t)U_H(t), \qquad U_H(0) = I_{\mathsf S},
\end{equation*}
where $H\colon[0,T]\to\operatorname{Herm}(\Sys)$ satisfies
\begin{equation}
  \norm{H(t)} \leq \alpha, \qquad \norm{H(t) - H(s)} \leq \beta|t - s|.
  \label{eq:review-hamiltonian-assumptions}
\end{equation}

We use the following block-encoding definition.
\begin{definition}[{\cite[Definition~43]{GSLW19Full}}]
Let $\mathsf R$ be an auxiliary register and let $\alpha>0$.
A unitary $V$ on $\mathcal H_{\mathsf R}\otimes\Sys$ block-encodes an operator
$X$ on $\Sys$ with normalization $\alpha$ if
\begin{equation*}
  (\bra0_{\mathsf R} \otimes I_{\mathsf S})V (\ket0_{\mathsf R} \otimes
      I_{\mathsf S}) = \frac{X}{\alpha}.
\end{equation*}
When $\alpha=1$, we say that $V$ block-encodes $X$.
\end{definition}

For a power of two $J$, let $\mathsf T$ be a $J$-dimensional time register.
Its basis state $\ket j_{\mathsf T}$ labels the sample time $t_j=jT/J$, and we
set $H_j=H(t_j)$.
We assume access to $H$ through $\HAMT_J$ as defined below, and, as in
\cite[Definition~3]{CGWZ26}, take each $O_j$ to be Hermitian.

\begin{definition}[{\cite[Definition 2]{LW18}}]
Let $H_0,\ldots,H_{J-1}\in\operatorname{Herm}(\Sys)$ satisfy
$\norm{H_j}\leq\alpha$, and let $\mathsf A$ be an $a$-qubit register.
The $\HAMT_J$ oracle for $\{H_j\}_{j=0}^{J-1}$ is
\begin{equation}
  \HAMT_J := \sum_{j = 0}^{J - 1} \ketbra{j}{j}_{\mathsf T} \otimes O_j,
  \label{eq:review-hamt}
\end{equation}
where each $O_j$ is a Hermitian unitary on $\mathcal H_{\mathsf A}\otimes\Sys$
satisfying
\begin{equation}
  (\bra0_{\mathsf A} \otimes I_{\mathsf S}) O_j (\ket0_{\mathsf A} \otimes
      I_{\mathsf S}) = \frac{H_j}{\alpha}.
  \label{eq:review-block-encoding}
\end{equation}
\end{definition}
Thus each $O_j$ block-encodes $H_j$ with normalization $\alpha$.
We henceforth write $\HAMT$ for $\HAMT_J$.

The $\HAMT_J$ model satisfies the following auxiliary-space lower bound.
\begin{proposition}[Auxiliary-space lower bound]
\label{prop:auxiliary-lower-bound}
Let $a\geq1$ and $0<\eps\leq\min\{1/2,\alpha T/8\}$.
Any algorithm that, given only $\HAMT_J$ in \eqref{eq:review-hamt} with
$t_j=jT/J$, approximates $U_H(T)$ to error $\eps$ for every $H$ satisfying
\eqref{eq:review-hamiltonian-assumptions} must have
\begin{equation*}
  J \geq \frac{\beta T^2}{16 \eps}.
\end{equation*}
It therefore requires
\begin{equation*}
  \Omega\! \left( a + \log\! \left(1 + \frac{\beta T^2}{\eps} \right) \right)
\end{equation*}
auxiliary qubits beyond $\mathsf S$.
\end{proposition}

\begin{proof}
When $\beta=0$, the bound on $J$ is immediate and the $a$-qubit register
$\mathsf A$ gives the auxiliary-space bound.
Hence assume $\beta>0$.
Set $\Delta=T/J$.
For $0\leq j<J$, define on $[j\Delta,(j+1)\Delta]$
\begin{equation*}
  h_0(t) = \min\{\beta(t - j \Delta), \beta((j + 1) \Delta - t), \alpha\}.
\end{equation*}
Each piece has absolute slope at most $\beta$, and adjacent pieces agree at
their endpoints.
Thus $h_0$ is $\beta$-Lipschitz on $[0,T]$, satisfies $0\leq h_0(t)\leq\alpha$,
and vanishes at every sample time.
Moreover,
\begin{equation*}
  \int_{j \Delta}^{(j + 1) \Delta}h_0(t) \dd{t} =
  \begin{cases}
    \beta \Delta^2/4, &\alpha \geq \beta \Delta/2, \\
    \alpha \Delta - \alpha^2/\beta, &\alpha < \beta \Delta/2.
  \end{cases}
\end{equation*}
In the first case, the integral is at least
$\frac14\min\{\beta\Delta^2,\alpha\Delta\}$.
In the second, $\min\{\beta\Delta^2,\alpha\Delta\}=\alpha\Delta$ and
$\alpha\Delta-\alpha^2/\beta>\alpha\Delta/2$.
Summing over the $J$ intervals therefore gives
\begin{equation}
  \int_0^T h_0(t) \dd{t} \geq \frac14 \min\!
      \left\{\frac{\beta T^2}{J}, \alpha T \right\}.
  \label{eq:auxiliary-lower-bound-area}
\end{equation}

Suppose that $J<\beta T^2/(16\eps)$.
By \eqref{eq:auxiliary-lower-bound-area} and $\alpha T\geq8\eps$, the integral
of $h_0$ is greater than $2\eps$.
Define
\begin{equation*}
  h(t) := \frac{2 \eps}{\int_0^T h_0(s)\, \dd{s}}\, h_0(t).
\end{equation*}
The prefactor is smaller than one, so $h$ is $\beta$-Lipschitz, bounded by
$\alpha$, zero at every sample time, and satisfies
\begin{equation*}
  \int_0^T h(t) \dd{t} = 2 \eps.
\end{equation*}
On a single-qubit system, consider the commuting Hamiltonians
\begin{equation*}
  H_+(t) = h(t)Z, \qquad H_-(t) = -h(t)Z.
\end{equation*}
For both Hamiltonians, $H_\pm(t_j)=0$.
Since $a\geq1$, choose every $O_j$ to apply a Pauli $X$ to one qubit of
$\mathsf A$ and the identity to all other qubits.
Its encoded block is zero because $\bra0X\ket0=0$, so the two Hamiltonians give
the same $\HAMT_J$ oracle.
Since both are proportional to $Z$ at all times,
\begin{equation*}
  U_{H_+}(T) = e^{-2i \eps Z}, \qquad U_{H_-}(T) = e^{2i \eps Z},
\end{equation*}
and hence
\begin{equation*}
  \norm{U_{H_+}(T) - U_{H_-}(T)} = 2 \sin(2 \eps) > 2 \eps.
\end{equation*}
Here the last inequality follows from $\sin x>x/2$ for $0<x\leq1$.
Since the oracle is identical in the two cases, the algorithm produces the same
approximation.
If both errors were at most $\eps$, the triangle inequality would give
\begin{equation*}
  \norm{U_{H_+}(T) - U_{H_-}(T)} \leq \eps + \eps = 2 \eps,
\end{equation*}
a contradiction.
This proves the bound on $J$.

By definition, $\HAMT_J$ acts on the $a$-qubit register $\mathsf A$ and the
$\log_2J$-qubit register $\mathsf T$, in addition to $\mathsf S$.
Hence any circuit using this oracle has at least $a+\log_2J$ auxiliary qubits
beyond $\mathsf S$.
Since $a\geq1$, the bound on $J$ implies
\begin{equation*}
  a + \log_2J = \Omega\! \left( a + \log\! \left(1 + \frac{\beta T^2}{\eps}
      \right) \right),
\end{equation*}
as claimed.
\end{proof}

\subsection{Cayley steps}

Set $w=\alpha T/J$ and define the Cayley steps
\begin{equation}
  U_j^C = \left(I_{\mathsf S} - \frac{iw}{2 \alpha}H_j \right)
      \left(I_{\mathsf S} + \frac{iw}{2 \alpha}H_j \right)^{-1}.
  \label{eq:review-cayley-step}
\end{equation}
Define their product by
\begin{equation}
  U_C := U_{J - 1}^C \cdots U_1^CU_0^C.
  \label{eq:review-cayley-product}
\end{equation}
For $\ket{\psi_0}\in\Sys$, define the intermediate states recursively by
\begin{equation}
  \ket{\psi_{j + 1}} := U_j^C \ket{\psi_j}, \qquad U_C \ket{\psi_0} =
      \ket{\psi_J}.
  \label{eq:review-cayley-states}
\end{equation}
The Cayley product approximates the target evolution with the following error.
\begin{lemma}[{\cite[Eq.~(13)]{CGWZ26}}]
\label{lem:review-cayley-approximation}
Suppose $H$ satisfies \eqref{eq:review-hamiltonian-assumptions}.
The Cayley product $U_C$ in \eqref{eq:review-cayley-product} satisfies
\begin{equation}
  \norm{U_C - U_H(T)} \leq \frac{\beta T^2}{2J} + \frac{(\alpha T)^3}{12J^2}.
  \label{eq:review-cayley-error}
\end{equation}
\end{lemma}

\subsection{One-query Cayley transducer}
We use the transducer framework introduced in \cite[Sec.~3.1]{BJY24}.
\begin{definition}
Let $\mathcal H_{\mathrm{pub}}$ and $\mathcal H_{\mathrm{priv}}$ be Hilbert
spaces, and let $U$ be a unitary on $\mathcal H_{\mathrm{pub}}$.
A unitary $S$ on $\mathcal H_{\mathrm{pub}}\oplus\mathcal H_{\mathrm{priv}}$ is
a transducer implementing $U$ if there is a linear map
$\Gamma\colon\mathcal H_{\mathrm{pub}}\to
\mathcal H_{\mathrm{priv}}$ such that
\begin{equation*}
  S(\ket \psi \oplus \Gamma \ket \psi) = U \ket \psi \oplus \Gamma \ket \psi
\end{equation*}
for every $\ket\psi\in\mathcal H_{\mathrm{pub}}$.
The vector $\Gamma\ket\psi$ is called the catalyst for $\ket\psi$.
\end{definition}

We now recall from \cite{CGWZ26} a one-query transducer implementing $U_C$.
Its public and private spaces are
\begin{equation*}
  \mathcal H_{\mathrm{pub}} :=
  \Sys, \qquad \Priv := \Time \otimes \Anc \otimes \Sys.
\end{equation*}
We represent $\Priv$ by the registers $\mathsf T\mathsf A\mathsf S$.
For $\ket\psi\in\Sys$ and $\ket v\in\Priv$, a flag qubit $\mathsf P$ represents
their direct sum by
\begin{equation*}
  \ket \psi_{\mathsf S} \oplus \ket v_{\mathsf{TAS}} \quad \longleftrightarrow
      \quad \ket0_{\mathsf{PTA}} \ket \psi_{\mathsf S} + \ket1_{\mathsf P} \ket
      v_{\mathsf{TAS}},
\end{equation*}
where
$\ket0_{\mathsf{PTA}} =\ket0_{\mathsf P}\ket0_{\mathsf T}\ket0_{\mathsf A}$.
Thus $\mathsf P=0$ and $\mathsf P=1$ denote the public and private summands,
respectively.

Define
\begin{equation}
  \iota_0 := \ket0_{\mathsf A} \otimes I_{\mathsf S}, \qquad \Pi_{\mathsf A, 0}
      := \ketbra{0}{0}_{\mathsf A} \otimes I_{\mathsf S},
  \label{eq:review-iota-projector}
\end{equation}
and the midpoint and local catalyst
\begin{equation}
  \ket{y_j} := \frac{\ket{\psi_j} + \ket{\psi_{j + 1}}}{2}, \qquad \ket{x_j} :=
      \sqrt{\frac w2}(I_{\mathsf A \mathsf S} + iO_j) \iota_0 \ket{y_j}.
  \label{eq:review-local-catalyst}
\end{equation}
Here $\iota_0$ appends $\ket0_{\mathsf A}$, and $\Pi_{\mathsf A,0}$ projects
$\mathsf A$ onto $\ket0_{\mathsf A}$.
Since $\ket{x_j}\in\Anc\otimes\Sys$, define $\Gamma\colon\Sys\to\Priv$ by
\begin{equation}
  \Gamma \ket{\psi_0} := \sum_{j = 0}^{J - 1} \ket j_{\mathsf T} \ket{x_j}.
  \label{eq:review-global-catalyst}
\end{equation}
As shown in \cref{prop:review-one-query-transducer}, $\Gamma\ket{\psi_0}$ is the
catalyst for input $\ket{\psi_0}$ of the transducer implementing $U_C$.

Let
\begin{equation}
  c = \frac{1 - w/2}{1 + w/2}, \qquad s = \frac{\sqrt{2w}}{1 + w/2},
  \label{eq:review-cs}
\end{equation}
and define the local update on $\Sys\oplus(\Anc\otimes\Sys)$ by
\begin{equation}
  R =
  \begin{pmatrix}
    cI_{\mathsf S} & -is \iota_0^\dagger \\
    s \iota_0 & i \bigl(c \Pi_{\mathsf A, 0} + I_{\mathsf A \mathsf S}
    - \Pi_{\mathsf A, 0} \bigr)
  \end{pmatrix}.
  \label{eq:review-local-update}
\end{equation}
The next lemma gives the update rule used at every time label.
\begin{lemma}[{\cite[Lemma~2]{CGWZ26}}]
\label{lem:review-local-cayley}
Let $R$ and $\ket{x_j}$ be defined by \eqref{eq:review-local-update} and
\eqref{eq:review-local-catalyst}, with $\ket{\psi_j}$ given by
\eqref{eq:review-cayley-states}.
Then $R$ is unitary and, for every $j$,
\begin{equation}
  R \bigl(\ket{\psi_j} \oplus O_j \ket{x_j} \bigr)
  = \ket{\psi_{j + 1}} \oplus \ket{x_j}.
  \label{eq:review-local-cayley-identity}
\end{equation}
Moreover,
\begin{equation}
  \norm{\ket{x_j}}^2 \leq w \norm{\ket{\psi_j}}^2.
  \label{eq:review-local-catalyst-bound}
\end{equation}
\end{lemma}

For $\ket\psi\in\Sys$ and $\ket{v_k}\in\Anc\otimes\Sys$, define $R_j$ on
$\Sys\oplus\Priv$ by
\begin{equation}
  R_j \left( \ket \psi \oplus \sum_{k = 0}^{J - 1} \ket k_{\mathsf T} \ket{v_k}
      \right) = \ket{\psi'} \oplus \left( \ket j_{\mathsf T} \ket{v_j'} +
      \sum_{k \ne j} \ket k_{\mathsf T} \ket{v_k} \right),
  \label{eq:review-global-update-definition}
\end{equation}
where
\begin{equation*}
  \ket{\psi'} \oplus \ket{v_j'} = R \bigl(\ket \psi \oplus \ket{v_j} \bigr).
\end{equation*}
Thus $R_j$ applies $R$ jointly to the public space $\mathcal H_{\mathrm{pub}}$
and the private subspace with $\mathsf T=j$, and acts as the identity on all
other private subspaces.
On $\Sys\oplus\Priv$, define
\begin{equation}
  S = S^\circ(I_{\mathrm{pub}} \oplus \HAMT), \qquad S^\circ
  = R_{J - 1} \cdots R_0,
  \label{eq:review-transducer}
\end{equation}
where $I_{\mathrm{pub}}$ is the identity on $\Sys$.
Applying \eqref{eq:review-local-cayley-identity} successively for
$j=0,\ldots,J-1$ gives the following transducer.

\begin{proposition}[{\cite[Proposition~3]{CGWZ26}}]
\label{prop:review-one-query-transducer}
The unitary $S$ on $\Sys\oplus\Priv$ defined in \eqref{eq:review-transducer}
uses one $\HAMT$ query and satisfies
\begin{equation}
  S \bigl(\ket \psi \oplus \Gamma \ket \psi \bigr)
  = U_C \ket \psi \oplus \Gamma \ket \psi
  \label{eq:review-transducer-identity}
\end{equation}
for every $\ket\psi\in\Sys$.
Moreover,
\begin{equation}
  \norm{\Gamma}^2 \leq \alpha T,
  \label{eq:review-global-catalyst-bound}
\end{equation}
where $U_C$ is defined in \eqref{eq:review-cayley-product}.
\end{proposition}

\subsection{Reuse, LCU, and amplification}

The transducer identity \eqref{eq:review-transducer-identity} assumes that the
catalyst $\Gamma\ket\psi$ is supplied.
For every integer $N\geq1$, the reuse construction of Belovs, Jeffery, and
Yolcu~\cite[Theorem~3.2]{BJY24} removes the need for a catalyst input by using
$N$ calls to $S$, which is the motivation of our simulation algorithm.
Let $P_N$ be the resulting operator on $\mathcal H_{\mathrm{pub}}$.
In \cref{sec:lcu}, we use a reuse register $\mathsf K$ to implement a unitary
$V_N$ that block-encodes $P_N$ with normalization $1$:
\begin{equation}
  (\bra0_{\mathsf{KPTA}} \otimes I_{\mathsf S})V_N (\ket0_{\mathsf{KPTA}}
      \otimes I_{\mathsf S}) = P_N.
  \label{eq:review-reuse-block}
\end{equation}
To describe the error, define $A\colon\Priv\to\Priv$ and
$C\colon\Priv\to\mathcal H_{\mathrm{pub}}$ by
\begin{equation}
  S(0 \oplus \ket v) = C \ket v \oplus A \ket v, \qquad \ket v \in \Priv.
  \label{eq:review-private-blocks}
\end{equation}
Set $G_N(z) := N^{-1}\sum_{\ell=0}^{N-1}z^\ell$.
The next lemma gives the reuse error.
\begin{lemma}[{\cite[Lemma~4]{CGWZ26}}]
\label{lem:review-reuse-error}
For every $N\geq1$, the operator $P_N$ in \eqref{eq:review-reuse-block}
satisfies
\begin{equation}
  U_C - P_N = C\, G_N(A) \Gamma.
  \label{eq:review-reuse-error}
\end{equation}
Here $\Gamma$ is defined in \eqref{eq:review-global-catalyst}.
\end{lemma}
For a single $N$, the general error bound
$\norm{U_C-P_N}\leq\norm{\Gamma}/\sqrt N$ decays only as $N^{-1/2}$.
To obtain the optimal precision dependence, we instead combine several values of
$N$.
For real coefficients $\{\lambda_N\}$ satisfying $\sum_N \lambda_N = 1$, set
\begin{equation}
  \widetilde U := \sum_N \lambda_N P_N, \qquad
      L_q := \sum_N |\lambda_N|.
  \label{eq:review-weighted-reuse}
\end{equation}
Summing \eqref{eq:review-reuse-error} with these coefficients gives
\begin{equation}
  U_C - \widetilde U = C \left(\sum_N \lambda_N G_N(A) \right) \Gamma.
  \label{eq:review-weighted-reuse-identity}
\end{equation}
Applying the linear-combination construction for block encodings
\cite[Lemma~52]{GSLW19Full} to the unitaries $V_N$, we obtain a block encoding
of $\widetilde U$ with normalization $L_q$.
We therefore seek coefficients that make the right-hand side of
\eqref{eq:review-weighted-reuse-identity} small while satisfying $L_q\leq2$.
The following choice has both properties.
\begin{proposition}[{\cite[Proposition~7]{CGWZ26}}]
\label{prop:review-combination}
For every positive integer $q$, define
\begin{align}
  \lambda_{2r} &:= -\frac rq\, 2^{-q} \binom qr,
  &&1 \leq r \leq q - 1,
  \label{eq:review-filter-coefficients-negative}\\
  \lambda_{2q + 2r} &:= \frac{q + r}{q}\, 2^{-q} \binom qr,
  &&1 \leq r \leq q,
  \label{eq:review-filter-coefficients}
\end{align}
and set all remaining coefficients to zero; in particular, $\lambda_{2q} = 0$.
These coefficients satisfy
\begin{equation}
  \sum_{N = 1}^{4q} \lambda_N = 1, \qquad L_q = \sum_{N = 1}^{4q}|\lambda_N| = 2
      - 2^{1 - q} < 2.
  \label{eq:review-coefficient-norm}
\end{equation}
The corresponding operator $\widetilde U$ in \eqref{eq:review-weighted-reuse}
satisfies
\begin{equation}
  \norm{U_C - \widetilde U} \leq \sqrt{\alpha T} \left(\frac{12e \alpha T}{q}
      \right)^q
  \label{eq:review-weighted-reuse-error}
\end{equation}
whenever $q\geq\alpha T$ and $J\geq q$.
\end{proposition}

The bound $L_q<2$ in \cref{prop:review-combination} allows us to increase the
normalization to $2$ by adding a term that contributes zero to the encoded
operator.
Let $\operatorname{PREP}$ prepare the resulting coefficient state, and let
$\operatorname{SELECT}$ apply $\operatorname{sgn}(\lambda_N)V_N$ according to
the value of the coefficient register.
The resulting LCU circuit is
\begin{equation}
  V = \operatorname{PREP}^\dagger\cdot \operatorname{SELECT} \cdot \operatorname{PREP}.
  \label{eq:review-lcu-circuit}
\end{equation}
Let $\mathsf W$ denote the combined auxiliary register of $V$.
Then $V$ block-encodes $\widetilde U$ with normalization $2$:
\begin{equation}
  (\bra0_{\mathsf W} \otimes I_{\mathsf S})V (\ket0_{\mathsf W} \otimes
      I_{\mathsf S}) = \frac{\widetilde U}{2}.
  \label{eq:review-lcu-block}
\end{equation}
If $\widetilde U=U_C$, then $V$ block-encodes the unitary $U_C$ with
normalization $2$.
Since $2=\csc(\pi/6)$, one step of oblivious amplitude
amplification~\cite{BerryEtAl14OAA}, equivalently the $n=3$ case of
\cite[Theorem~28]{GSLW19Full}, recovers $U_C$ using $V,V^\dagger,V$.
Define the reflection
\begin{equation}
  \mathcal R_{\mathsf W} := 2(\ketbra{0}{0}_{\mathsf W} \otimes I_{\mathsf S})
      -I_{\mathsf W \mathsf S}
  \label{eq:review-clean-reflection}
\end{equation}
and define the amplified circuit
\begin{equation}
  U_{\mathrm{sim}} := -V \mathcal R_{\mathsf W}V^\dagger
      \mathcal R_{\mathsf W}V.
  \label{eq:review-amplified-circuit}
\end{equation}
Since $\widetilde U$ only approximates $U_C$, we use the following robust bound,
which also accounts for the part of the output with $\mathsf W\ne0$.
\begin{lemma}[{\cite[Lemma~8]{CGWZ26}}]
\label{lem:review-robust-amplification}
Let $U$ be a unitary on $\Sys$, and let $V$ on
$\mathcal H_{\mathsf W}\otimes\Sys$ block-encode an operator $\widetilde U$ on
$\Sys$ with normalization $2$.
Set $\eta:=\norm{\widetilde U-U}$.
If $\eta\leq1/8$, then, for every $\ket\psi\in\Sys$, the circuit in
\eqref{eq:review-amplified-circuit} satisfies
\begin{equation}
  \norm{ U_{\mathrm{sim}}(\ket0_{\mathsf W} \ket \psi) -
      \ket0_{\mathsf W}U \ket \psi}
      \leq 3 \eta.
  \label{eq:review-amplification-error}
\end{equation}
\end{lemma}

\Cref{alg:review-cgwz} restates \cite[Algorithm~2]{CGWZ26} in the notation used
here.
\Cref{sec:implementation,sec:lcu} implement the same $S^\circ$ and block-encode
the same operators $P_N$, while making the gate costs of $\operatorname{PREP}$
and $\operatorname{SELECT}$ explicit.

\begin{algorithm}[H]
\caption{Time-dependent Hamiltonian simulation from
\cite[Algorithm~2]{CGWZ26}}
\label{alg:review-cgwz}
\small
\begin{algorithmic}[1]
\Require Access to $\HAMT$ and a target error
         $0<\eps\leq1/2$ with $\alpha T>\eps$.
\Ensure A circuit $U_{\mathrm{sim}}$ approximating $U_H(T)$ to error
        at most $\eps$.
\State Choose the smallest integer $q\geq\max\{1,\alpha T\}$ such that
       $\sqrt{\alpha T}(12e\alpha T/q)^q\leq\eps/12$.
\State Choose the smallest power of two $J$ satisfying
       \begin{equation*}
         J \geq \max
             \left\{ 2, q, \frac{\beta T^2}{\eps},
             \frac{(\alpha T)^{3/2}}{\sqrt{3 \eps}} \right\}.
       \end{equation*}
\State Set $t_j=jT/J$ and define the Cayley product
       $U_C$ in \eqref{eq:review-cayley-product}.
\State Construct the one-query transducer
       $S$ in \eqref{eq:review-transducer} for $U_C$.
\State For $N=1,\ldots,4q$, let $V_N$ be the $N$-call reuse unitary that
       block-encodes $P_N$ as in \eqref{eq:review-reuse-block}.
\State Choose $\lambda_1,\ldots,\lambda_{4q}$ from
       \eqref{eq:review-filter-coefficients-negative}--
       \eqref{eq:review-filter-coefficients} and form
       $\widetilde U = \sum_N \lambda_N P_N$.
\State Construct $V$ in \eqref{eq:review-lcu-circuit}, which block-encodes
       $\widetilde U$ with normalization $2$.
\State Apply the one-step amplification circuit
       $U_{\mathrm{sim}}$ in \eqref{eq:review-amplified-circuit}.
\State \Return $U_{\mathrm{sim}}$.
\end{algorithmic}
\end{algorithm}

\section{Gate-efficient implementation}
\label{sec:implementation}

We first state the resulting complexity bound and then give a gate-efficient
implementation of $S^\circ$.
\begin{theorem}[Gate-efficient query-optimal simulation]
\label{thm:gate-efficient}
Let $H$ satisfy \eqref{eq:review-hamiltonian-assumptions}, let
$0<\eps\leq1/2$ and $\alpha T>\eps$, and assume access to $\HAMT$ in
\eqref{eq:review-hamt} with an $a$-qubit block-encoding register.  There exist
an auxiliary register $\mathsf W$ and a unitary circuit $U_{\mathrm{sim}}$ such
that, for every state $\ket{\psi_0}_{\mathsf S}$,
\begin{equation*}
  \norm{ U_{\mathrm{sim}} \bigl(\ket0_{\mathsf W} \otimes
      \ket{\psi_0} \bigr) - \ket0_{\mathsf W} \otimes U_H(T)
      \ket{\psi_0}}
      \leq \eps.
\end{equation*}
The parameter $q$ chosen in \cref{alg:review-cgwz} satisfies
\begin{equation*}
  q = \bigO\! \left( \alpha T + \frac{\log(1/\eps)}
      {\log\! \left(e + \log(1/\eps)/(\alpha T) \right)} \right),
\end{equation*}
and the circuit uses $\bigO(q)$ $\HAMT$ queries.  Furthermore, the circuit uses
\begin{equation*}
  n_{\mathrm{anc}} = \bigO\! \left( a + \log\! \left(1 + \frac{T(\alpha + \beta T)}{\eps} \right)
      \right)
\end{equation*}
auxiliary qubits and $\bigO(qn_{\mathrm{anc}})$ one- and two-qubit gates.
\end{theorem}

\subsection{Dyadic factorization}
\label{sec:dyadic}

The direct implementation of $S^\circ=R_{J-1}\cdots R_0$ applies $J$ updates,
one for each sampled time.
We show that the same unitary factors into an input basis change on $\mathsf T$,
a rotation and a phase gate on $\mathsf P$, and an output basis change on
$\mathsf T$.
Each basis change contains one controlled one-qubit gate per time qubit.

Since every $R_j$, and hence $S^\circ$, acts as the identity on $\mathsf S$, we
suppress this identity throughout the subsection and work on
$\mathsf P\mathsf T\mathsf A$.

Write $J=2^m$, and denote the qubits of $\mathsf T$ by
$\mathsf T_0,\ldots,\mathsf T_{m-1}$.
On $\ket j_{\mathsf T}$, the qubit $\mathsf T_\ell$ stores $j_\ell$ in the
binary expansion
\begin{equation*}
  j = \sum_{\ell = 0}^{m - 1}2^\ell j_\ell, \qquad j_\ell \in \{0, 1\}.
\end{equation*}
Throughout this section, $J\geq\alpha T>0$, so $0<w=\alpha T/J\leq1$.
Hence the numbers in \eqref{eq:review-cs} satisfy
\begin{equation*}
  0 < c < 1, \qquad s = \sqrt{1 - c^2}.
\end{equation*}
For $0\leq r\leq1$, define the one-qubit gates
\begin{equation*}
  \operatorname{Rot}(r) :=
  \begin{pmatrix}
    r & - \sqrt{1 - r^2} \\
    \sqrt{1 - r^2} & r
  \end{pmatrix},
\end{equation*}
and
\begin{equation*}
  M_{\mathrm{in}}(r)
  = \frac1{\sqrt{1 + r^2}}
  \begin{pmatrix}r & 1 \\ 1 & -r \end{pmatrix},
  \qquad
  M_{\mathrm{out}}(r)
  = \frac1{\sqrt{1 + r^2}}
  \begin{pmatrix}1 & r \\ r & -1 \end{pmatrix}.
\end{equation*}
Both $M_{\mathrm{in}}(r)$ and $M_{\mathrm{out}}(r)$ are one-qubit unitaries.

For $\star\in\{\mathrm{in},\mathrm{out}\}$, define $W_J^\star$ as a circuit on
$\mathsf T$.
For $\ell=m-1,m-2,\ldots,0$, apply
\begin{equation*}
  M_\star(c^{2^\ell}) \quad \text{to } \mathsf T_\ell
\end{equation*}
when
\begin{equation*}
  \mathsf T_0 = \cdots = \mathsf T_{\ell - 1} = 0.
\end{equation*}
The gate on $\mathsf T_0$ has no time-register controls.
The circuit $W_J^\star$ contains exactly $m$ one-qubit gates with the controls
specified above.
For $J=1$, set $W_1^\star=I_{\mathsf T}$.
\Cref{fig:dyadic-time-circuit} shows this circuit for $J=8$ and highlights the
decomposition used in the induction proof below.
\begin{figure}[ht]
  \centering
  \begin{quantikz}[row sep=0.35cm, column sep=0.55cm]
    \lstick{$\mathsf T_2$}
    & \gate{M_\star(c^4)}
    & \qw\gategroup[
    wires=3,
    steps=2,
    style={dashed,rounded corners,inner xsep=4pt,inner ysep=3pt},
    background,
    label style={label position=below,anchor=north,yshift=-0.15cm}
    ]{$I_{\mathsf T_2}\otimes W_4^\star$}
    & \qw \\
    \lstick{$\mathsf T_1$}
    & \octrl{-1}
    & \gate{M_\star(c^2)}
    & \qw \\
    \lstick{$\mathsf T_0$}
    & \octrl{-2}
    & \octrl{-1}
    & \gate{M_\star(c)}
  \end{quantikz}
  \caption{The circuit for $W_8^\star$.
    The dashed box applies
    $W_4^\star$ to $\mathsf T_1\mathsf T_0$ and leaves $\mathsf T_2$
    unchanged.}
  \label{fig:dyadic-time-circuit}
\end{figure}
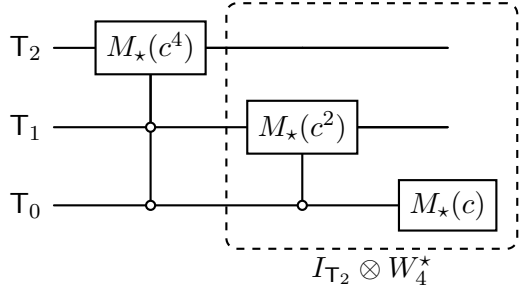

Extend $S^\circ$ from the direct-sum space to $\mathsf P\mathsf T\mathsf A$ by
letting it act as the identity on the orthogonal complement.
On $\mathsf P\mathsf T\mathsf A$, let $\widehat W_J^\star$ apply $W_J^\star$ to
$\mathsf T$ when $\mathsf P=1$ and $\mathsf A=0$, and let
$\widehat{\operatorname{Rot}}(c^J)$ apply $\operatorname{Rot}(c^J)$ to
$\mathsf P$ when $\mathsf T=0$ and $\mathsf A=0$.
Both operators are the identity otherwise.
Also set
\begin{equation*}
  \Phi_{\mathsf P} := \bigl(\ketbra{0}{0}_{\mathsf P} + i
      \ketbra{1}{1}_{\mathsf P} \bigr) \otimes I_{\mathsf T \mathsf A}.
\end{equation*}
We will prove
\begin{equation}
  S^\circ = \widehat W_J^{\mathrm{out}} \widehat{\operatorname{Rot}}(c^J)
      \Phi_{\mathsf P} (\widehat W_J^{\mathrm{in}})^\dagger.
  \label{eq:dyadic-full-factorization}
\end{equation}

Let $\Phi_j$ multiply by $i$ on the private states with $\mathsf T=j$ and act as
the identity otherwise.
Define
\begin{equation*}
  \begin{aligned}
    \ket{\mathrm{pub}}
    &:= \ket0_{\mathsf P} \ket0_{\mathsf T} \ket0_{\mathsf A}, \\
    \ket j
    &:= \ket1_{\mathsf P} \ket j_{\mathsf T} \ket0_{\mathsf A}.
  \end{aligned}
\end{equation*}
Define $g_j$ on these states by
\begin{equation}
  \begin{aligned}
    g_j \ket{\mathrm{pub}}
    &= c \ket{\mathrm{pub}} + s \ket j, \\
    g_j \ket j
    &= -s \ket{\mathrm{pub}} + c \ket j, \\
    g_j \ket k
    &= \ket k \qquad(k \ne j).
  \end{aligned}
  \label{eq:dyadic-givens-action}
\end{equation}
Let $g_j$ leave every remaining basis state of $\mathsf P\mathsf T\mathsf A$
unchanged.
Comparing \eqref{eq:dyadic-givens-action} with \eqref{eq:review-local-update}
gives $R_j = g_j \Phi_j$.
Since $\Phi_j g_k = g_k \Phi_j$ for $j \ne k$,
\begin{equation}
  S^\circ = C_J \Phi_{\mathsf P}, \qquad C_J := g_{J - 1} \cdots g_0,
  \label{eq:dyadic-remove-phase}
\end{equation}
because $\Phi_{J - 1} \cdots \Phi_0 = \Phi_{\mathsf P}$.

The operator $C_J$ acts as the identity on the orthogonal complement of
$\operatorname{span}\{\ket{\mathrm{pub}},\ket0,\ldots, \ket{J-1}\}$.
It therefore suffices to work on this subspace in the ordered basis
$(\ket{\mathrm{pub}},\ket0,\ldots,\ket{J-1})$.
For $0\leq r\leq1$, write $\operatorname{Rot}(r)_{\mathrm{pub},j}$ for
$\operatorname{Rot}(r)$ on $\operatorname{span}\{\ket{\mathrm{pub}},\ket j\}$,
in the ordered basis $(\ket{\mathrm{pub}},\ket j)$, and the identity on the
remaining time-label basis states.

The following lemma proves the required factorization.

\begin{lemma}[Dyadic factorization]
\label{lem:dyadic-factorization}
For every power of two $J$, on
$\operatorname{span}\{\ket{\mathrm{pub}},\ket0,\ldots,\ket{J-1}\}$, the circuits
$W_J^{\mathrm{in}}$ and $W_J^{\mathrm{out}}$ satisfy
\begin{equation}
  (\widehat W_J^{\mathrm{out}})^\dagger C_J \widehat W_J^{\mathrm{in}} =
      \operatorname{Rot}(c^J)_{\mathrm{pub}, 0}.
  \label{eq:dyadic-factorization}
\end{equation}
Thus $\widehat W_J^\star\ket{\mathrm{pub}} =\ket{\mathrm{pub}}$, while its
action on the ordered basis $(\ket0,\ldots,\ket{J-1})$ is the time-register
circuit $W_J^\star$.
\end{lemma}

\begin{proof}
For $J=1$, \eqref{eq:dyadic-factorization} is \eqref{eq:dyadic-givens-action}.
Let $J=2L$ and assume the lemma for $L$.
Set
\begin{equation*}
  r = c^L, \qquad t = \sqrt{1 - r^2}.
\end{equation*}
For $J=2L$, every time label has a unique decomposition
\begin{equation*}
  j = bL + k, \qquad b \in \{0, 1\}, \quad 0 \leq k < L.
\end{equation*}
Accordingly,
\begin{equation*}
  \ket j_{\mathsf T} = \ket b_{\mathsf T_{m - 1}} \ket
      k_{\mathsf T_{m - 2} \cdots \mathsf T_{0}}.
\end{equation*}
The first operation in $W_{2L}^\star$ applies $M_\star(r)$ to $\mathsf T_{m-1}$
when $k=0$.
It therefore mixes $\ket0$ and $\ket L$ and leaves all other time-label basis
states unchanged.
The remaining gates act on $\mathsf T_0,\ldots,\mathsf T_{m-2}$ and form
$W_L^\star$, independently of $b$.
Their product is $I_{\mathsf T_{m-1}}\otimes W_L^\star$.
Thus $W_{2L}^\star$ first applies the controlled $M_\star(r)$ gate and then
$I_{\mathsf T_{m-1}}\otimes W_L^\star$, as illustrated for $J=8$ in
\cref{fig:dyadic-time-circuit}.

For $b\in\{0,1\}$, let $V_b^\star$ fix $\ket{\mathrm{pub}}$, act as $W_L^\star$
on the ordered basis
\begin{equation}
  (\ket{bL}, \ket{bL + 1}, \ldots, \ket{bL + L - 1}),
  \label{eq:dyadic-half-basis}
\end{equation}
and fix the states in the other half of the time-label range.
Since $C_{2L}=(g_{2L-1}\cdots g_L)(g_{L-1}\cdots g_0)$, the induction hypothesis
on $\operatorname{span}\{\ket{\mathrm{pub}},\ket0,\ldots, \ket{L-1}\}$ gives
\begin{equation}
  (V_0^{\mathrm{out}})^\dagger (g_{L - 1} \cdots g_0) V_0^{\mathrm{in}} =
      \operatorname{Rot}(r)_{\mathrm{pub}, 0}.
  \label{eq:dyadic-first-half}
\end{equation}
Identify the ordered basis in \eqref{eq:dyadic-half-basis} for $b=1$ with
$(\ket0,\ldots,\ket{L-1})$ by replacing each local time-label index $j$, where
$0\leq j<L$, with $L+j$.
Under this identification, the induction hypothesis replaces $g_j$ by $g_{L+j}$
and $V_0^\star$ by $V_1^\star$.
Hence
\begin{equation}
  (V_1^{\mathrm{out}})^\dagger (g_{2L - 1} \cdots g_L) V_1^{\mathrm{in}} =
      \operatorname{Rot}(r)_{\mathrm{pub}, L}.
  \label{eq:dyadic-second-half}
\end{equation}
The operators $V_0^\star$ commute with $g_L,\ldots,g_{2L-1}$, while the
operators $V_1^\star$ commute with $g_0,\ldots,g_{L-1}$.
Moreover, $V_0^\star$ and $V_1^\star$ commute because they act on disjoint
time-label ranges and both fix $\ket{\mathrm{pub}}$.
Multiplying \eqref{eq:dyadic-second-half} by \eqref{eq:dyadic-first-half}
therefore gives
\begin{equation}
  (V_0^{\mathrm{out}}V_1^{\mathrm{out}})^\dagger C_{2L}
  V_1^{\mathrm{in}}V_0^{\mathrm{in}} =
  \operatorname{Rot}(r)_{\mathrm{pub}, L}
  \operatorname{Rot}(r)_{\mathrm{pub}, 0}.
  \label{eq:dyadic-induction-split}
\end{equation}
The two rotations on the last line act nontrivially only on
$\operatorname{span}\{\ket{\mathrm{pub}},\ket0,\ket L\}$.
In the ordered basis $(\ket{\mathrm{pub}},\ket0,\ket L)$, their product is
\begin{equation}
  \begin{pmatrix}
    r & 0 & -t \\ 0 & 1 & 0 \\ t & 0 & r
  \end{pmatrix}
  \begin{pmatrix}
    r & -t & 0 \\ t & r & 0 \\ 0 & 0 & 1
  \end{pmatrix} =
  \begin{pmatrix}
    r^2 & -rt & -t \\
    t & r & 0 \\
    rt & -t^2 & r
  \end{pmatrix}.
  \label{eq:dyadic-three-dimensional-product}
\end{equation}
For either value of $\star$, the product $V_1^\star V_0^\star$ fixes
$\ket{\mathrm{pub}}$ and applies $I_{\mathsf T_{m-1}}\otimes W_L^\star$ to the
states $\ket j$.
On $\operatorname{span}\{\ket0,\ket L\}$, the controlled gates have matrices
$M_{\mathrm{in}}(r)$ and $M_{\mathrm{out}}(r)^\dagger$ in the ordered basis
$(\ket0,\ket L)$.
Thus \eqref{eq:dyadic-induction-split} reduces the induction step to
\begin{equation}
  \begin{aligned}
    &\begin{pmatrix}1 & 0 \\ 0 & M_{\mathrm{out}}(r)^\dagger \end{pmatrix}
    \begin{pmatrix}
      r^2 & -rt & -t \\
      t & r & 0 \\
      rt & -t^2 & r
    \end{pmatrix}
    \begin{pmatrix}1 & 0 \\ 0 & M_{\mathrm{in}}(r) \end{pmatrix} =
    \begin{pmatrix}
      r^2 & - \sqrt{1 - r^4} & 0 \\
      \sqrt{1 - r^4} & r^2 & 0 \\
      0 & 0 & 1
    \end{pmatrix}.
  \end{aligned}
  \label{eq:dyadic-three-dimensional-identity}
\end{equation}
The right-hand side is the matrix of $\operatorname{Rot}(r^2)_{\mathrm{pub},0}$
in the ordered basis $(\ket{\mathrm{pub}},\ket0,\ket L)$.
Since $r^2=c^{2L}$, this is \eqref{eq:dyadic-factorization} for $2L$.
\end{proof}

\begin{remark}
Equation~\eqref{eq:dyadic-factorization} is a cosine--sine decomposition of
$C_J$ with block sizes $1$ and $J$~\cite{PaigeWei94}.
Its two basis changes are given explicitly by the circuits $W_J^{\mathrm{in}}$
and $W_J^{\mathrm{out}}$.
\end{remark}

On the subspace in \cref{lem:dyadic-factorization},
$\widehat{\operatorname{Rot}}(c^J)$ agrees with
$\operatorname{Rot}(c^J)_{\mathrm{pub},0}$.
Rearranging \eqref{eq:dyadic-factorization} therefore gives
\begin{equation}
  C_J = \widehat W_J^{\mathrm{out}} \widehat{\operatorname{Rot}}(c^J) (\widehat
      W_J^{\mathrm{in}})^\dagger.
  \label{eq:dyadic-factorization-rearranged}
\end{equation}
Both sides act as the identity on the orthogonal complement of the subspace in
\cref{lem:dyadic-factorization}, so the equality holds on
$\mathsf P\mathsf T\mathsf A$.
Since $\Phi_{\mathsf P}$ commutes with $\widehat W_J^{\mathrm{in}}$,
\eqref{eq:dyadic-remove-phase} gives
\begin{align*}
  S^\circ
  &= C_J \Phi_{\mathsf P}\\
  &= \widehat W_J^{\mathrm{out}}
  \widehat{\operatorname{Rot}}(c^J) \Phi_{\mathsf P}
  (\widehat W_J^{\mathrm{in}})^\dagger,
\end{align*}
which proves \eqref{eq:dyadic-full-factorization}.

\subsection{Controlled implementation}

Tensoring each factor in \eqref{eq:dyadic-full-factorization} with
$I_{\mathsf S}$ gives the factorization of $S^\circ$ on
$\mathsf P\mathsf T\mathsf A\mathsf S$.
We use it to implement a controlled application of $S^\circ$.
Let $\mathsf Z$ denote the control qubit.
In $\widehat W_J^\star$, the gate on $\mathsf T_\ell$ is applied when
\begin{equation*}
  \mathsf Z = 1, \qquad \mathsf P = 1, \qquad \mathsf A = 0, \qquad \mathsf T_0
      = \cdots = \mathsf T_{\ell - 1} = 0.
\end{equation*}
Compute $f_0,\ldots,f_{m-1}$ into work qubits:
\begin{equation}
  f_0 = [\mathsf Z = 1]\wedge[\mathsf P = 1]\wedge[\mathsf A = 0], \qquad
      f_{\ell + 1} = f_\ell \wedge[\mathsf T_\ell = 0] \quad(0 \leq \ell < m -
      1).
  \label{eq:dyadic-control-bits}
\end{equation}
After computing $f_0,\ldots,f_{m-1}$, proceed from $\ell=m-1$ down to $0$: apply
$M_\star(c^{2^\ell})$ to $\mathsf T_\ell$, controlled by $f_\ell$.
Since $f_\ell$ does not depend on the target $\mathsf T_\ell$, reverse its
computation after the gate.

The factor $\widehat{\operatorname{Rot}}(c^J)$ requires a different control
because its target is $\mathsf P$.
Compute
\begin{equation}
  u = [\mathsf Z = 1]\wedge[\mathsf T = 0]\wedge[\mathsf A = 0],
  \label{eq:dyadic-central-control}
\end{equation}
apply $\operatorname{Rot}(c^J)$ to $\mathsf P$ controlled by $u$, and erase $u$.
The factor $\Phi_{\mathsf P}$ is one phase gate on $\mathsf P$, controlled by
$\mathsf Z$.

We record the resulting circuit cost for later use in SELECT\@.
\begin{proposition}[Controlled implementation of $S^\circ$]
\label{prop:gate-efficient-update-product}
Controlled implementations of $S^\circ$ and $(S^\circ)^\dagger$ each use
$\bigO(a+\log J+1)$ one- and two-qubit gates and $\bigO(a+\log J+1)$ auxiliary
qubits initialized and returned to zero.
\end{proposition}

\begin{proof}
For either $W_J^\star$, computing and erasing $f_0$ uses two $(a+2)$-controlled
NOTs.
Computing and erasing $f_1,\ldots,f_{m-1}$ uses $2(m-1)$ three-qubit Toffoli
gates.
Thus each of $(\widehat W_J^{\mathrm{in}})^\dagger$ and
$\widehat W_J^{\mathrm{out}}$, including its $m$ controlled one-qubit gates,
costs $\bigO(a+m)$ gates.

Computing and erasing $u$ uses two $(a+m+1)$-controlled NOTs.
By \cite[Corollary~7.4]{BarencoEtAl95}, a $k$-controlled NOT has an
$\bigO(k)$-gate decomposition using one reusable work qubit.
The gates $\operatorname{Rot}(c^J)$ controlled by $u$ and $\Phi_{\mathsf P}$
controlled by $\mathsf Z$ contribute $\bigO(1)$ gates.
Hence the total gate count is $\bigO(a+m+1)$.
The $m+1$ qubits storing $f_0,\ldots,f_{m-1},u$, together with one reusable work
qubit, suffice and are returned to zero.
\end{proof}

By \eqref{eq:review-transducer}, a controlled application of $S$ is obtained
from controlled applications of $S^\circ$ and $I_{\mathrm{pub}}\oplus\HAMT$.

\section{Implementation of PREP and SELECT}
\label{sec:lcu}

We now implement the $\operatorname{PREP}$ and $\operatorname{SELECT}$ circuits
in \eqref{eq:review-lcu-circuit}.
Each controlled application of $S^\circ$ in $\operatorname{SELECT}$ uses the
implementation from \cref{prop:gate-efficient-update-product}.
Use $b$-qubit registers $\mathsf B$ and $\mathsf K$, where
\begin{equation}
  b := \left \lceil \log_2(4q + 2) \right \rceil.
  \label{eq:lcu-register-width}
\end{equation}
The register $\mathsf B$ stores $N$, while $\mathsf K$ stores the reuse label
$\ell$.

\subsection{Coefficient preparation}

The nonzero coefficients in \eqref{eq:review-filter-coefficients-negative} and
\eqref{eq:review-filter-coefficients} have support
\begin{equation}
  \mathcal I_q = \{2r:1 \leq r \leq q - 1\} \cup\{2q + 2r:1 \leq r \leq q\}.
  \label{eq:supported-reuse-lengths}
\end{equation}
Thus $|\mathcal I_q|=2q-1$, and \eqref{eq:review-coefficient-norm} gives
$L_q=2-2^{1-q}$.
Let $\ket\perp_{\mathsf B}:=\ket{4q+1}_{\mathsf B}$.

Choose $\operatorname{PREP}$ such that
\begin{equation}
  \operatorname{PREP} \ket0_{\mathsf B} = \sum_{N \in \mathcal I_q}
      \sqrt{\frac{|\lambda_N|}{2}} \ket N_{\mathsf B} + \sqrt{\frac{2 - L_q}{2}}
      \ket \perp_{\mathsf B}.
  \label{eq:coefficient-state}
\end{equation}
Since $\sum_{N\in\mathcal I_q}|\lambda_N|=L_q$, the last term makes the state in
\eqref{eq:coefficient-state} normalized.
In $\operatorname{SELECT}$,
\begin{equation*}
  \ket \perp_{\mathsf B} \ket0_{\mathsf P} \longmapsto \ket \perp_{\mathsf B}
      \ket1_{\mathsf P},
\end{equation*}
so it contributes zero to the block in \eqref{eq:review-lcu-block}.
The construction of \cite{MottonenEtAl05} prepares an arbitrary state on $b$
qubits with $\bigO(2^b)$ one- and two-qubit gates.
Equation~\eqref{eq:lcu-register-width} gives $2^b<2(4q+2)=\bigO(q)$, so
$\operatorname{PREP}$ costs $\bigO(q)$ gates.

\subsection{The SELECT circuit}

We implement $\operatorname{SELECT}$ following \cite[Algorithm~1]{CGWZ26}, with
the reuse unitary $V_N$ defined below using $\operatorname{INC}$.
For every state $\ket\phi$ on $\mathsf{KPTAS}$, the required action is
\begin{equation}
  \operatorname{SELECT} \bigl(\ket N_{\mathsf B} \ket \phi \bigr) =
      \operatorname{sgn}(\lambda_N) \ket N_{\mathsf B}V_N \ket \phi, \qquad N
      \in \mathcal I_q.
  \label{eq:select-action}
\end{equation}
For each $N\in\mathcal I_q$, choose a unitary $F_N$ on $\mathsf K$ satisfying
\begin{equation}
  F_N \ket0_{\mathsf K} = \frac1{\sqrt N} \sum_{\ell = 0}^{N - 1} \ket
      \ell_{\mathsf K}.
  \label{eq:review-reuse-preparation}
\end{equation}
Let $\operatorname{INC}$ denote increment modulo the dimension of the $b$-qubit
register $\mathsf K$:
\begin{equation*}
  \operatorname{INC} \ket y_{\mathsf K} = \ket{(y + 1) \bmod 2^b}_{\mathsf K}.
\end{equation*}
The modulus $2^b$ is independent of the value stored in $\mathsf B$.

For each fixed $N$, define $V_N$ by the following three-step circuit:
\begin{enumerate}
 \item[(i)] Apply $F_N$ when $\mathsf P=0$.
 \item[(ii)] For $\ell=0,\ldots,N-1$, apply, in order, $\HAMT$ when
       $\mathsf P=1$, $S^\circ$ when $\mathsf K=\ell$, and
       $\operatorname{INC}$ when $\mathsf P=1$.
 \item[(iii)] Apply $F_N^\dagger$ when $\mathsf P=0$.
\end{enumerate}

The circuit in \cite[Eq.~(27)]{CGWZ26} has the same steps~(i) and~(iii), but in
step~(ii) it replaces $\operatorname{INC}$ by
\begin{equation*}
  \ket \ell_{\mathsf K} \longmapsto \ket{(\ell + 1) \bmod N}_{\mathsf K}, \qquad
      0 \leq \ell < N,
\end{equation*}
and leaves the remaining basis states of $\mathsf K$ unchanged.
We now show that replacing this update by $\operatorname{INC}$ does not change
the block in \eqref{eq:review-reuse-block}.
Apply both circuits to $\ket0_{\mathsf{KPTA}}\ket\psi_{\mathsf S}$.
Since $N\leq4q\leq2^b-2$, we have $\ell+1<2^b$ for $0\leq\ell<N$, and
\begin{equation*}
  \ket \ell_{\mathsf K} \ket1_{\mathsf P} \longmapsto \ket{\ell + 1}_{\mathsf K}
      \ket1_{\mathsf P}, \qquad 0 \leq \ell < N.
\end{equation*}
For $\ell<N-1$, both updates move every $\mathsf P=1$ component produced in
iteration $\ell$ from $\mathsf K=\ell$ to $\mathsf K=\ell+1$.
Thus the $S^\circ$ call in iteration $\ell+1$ acts on the same component in both
circuits.
They differ only after iteration $N-1$: $\operatorname{INC}$ moves the final
component with $\mathsf P=1$ to $\mathsf K=N$, whereas the update modulo $N$
moves it to $\mathsf K=0$.
In either circuit this component has $\mathsf P=1$, and no further iteration
acts on it.
Moreover, $F_N^\dagger$ acts only when $\mathsf P=0$.
Hence the bra $\bra0_{\mathsf{KPTA}}$ in \eqref{eq:review-reuse-block}
annihilates the only difference between the two outputs.
The two circuits therefore block-encode the same operator $P_N$.

The construction of \cite[Algorithm~1 and Section~2.5]{ShuklaVedula24}
implements both $F_N$ and $F_N^\dagger$, without ancilla qubits, using
$\bigO(\log(N+1))$ one- and two-qubit gates.
To condition either operation on $\mathsf B=N$ and $\mathsf P=0$, compute and
erase the condition with two additional $(b+1)$-controlled NOTs and control each
gate of the chosen circuit on the resulting work qubit.
This adds $\bigO(b)$ gates.
Using a $b$-qubit register initialized to zero as the addend and the control bit
as the incoming carry, the modulo-$2^b$ adder of
\cite[Sections~4.1--4.2]{CuccaroEtAl04} implements a controlled application of
$\operatorname{INC}$ with $\bigO(b)$ gates and returns the addend register to
zero.

We now implement \eqref{eq:select-action} for all values stored in $\mathsf B$.
By \eqref{eq:review-filter-coefficients-negative} and
\eqref{eq:review-filter-coefficients},
\begin{equation*}
  \operatorname{sgn}(\lambda_N)
  = \begin{cases}
    -1, &2 \leq N \leq 2q - 2, \\
    + 1, &2q + 2 \leq N \leq 4q,
  \end{cases}
  \qquad N \in \mathcal I_q.
\end{equation*}
The circuit is implemented in the following order.
\begin{enumerate}
 \item Apply a phase $-1$ when $\mathsf B$ is even and
       $2\leq\mathsf B\leq2q-2$.
       For each $N\in\mathcal I_q$, apply $F_N$ to
       $\mathsf K$, controlled on $\mathsf B=N$ and $\mathsf P=0$.
 \item For $\ell=0,\ldots,4q-1$, compute
       \begin{equation*}
         h_\ell = [\mathsf B \in \mathcal I_q]\wedge[\mathsf B > \ell].
       \end{equation*}
       Controlled on $h_\ell$, apply, in order,
       $\HAMT$ when $\mathsf P=1$, $S^\circ$ when
       $\mathsf K=\ell$, and $\operatorname{INC}$ when
       $\mathsf P=1$.
       Uncompute the work qubits before the next iteration.
 \item For each $N\in\mathcal I_q$, apply $F_N^\dagger$ to
       $\mathsf K$, controlled on $\mathsf B=N$ and $\mathsf P=0$.
 \item If $\mathsf B=\perp$, flip $\mathsf P$.
\end{enumerate}

For a fixed $N\in\mathcal I_q$,
\begin{equation*}
  h_\ell = [N > \ell]
  = \begin{cases}
    1, &0 \leq \ell < N, \\
    0, &N \leq \ell < 4q.
  \end{cases}
\end{equation*}
Thus item~2 of the SELECT circuit executes step~(ii) in the definition of $V_N$.
Items~1 and~3 execute steps~(i) and~(iii), while item~1 also applies
$\operatorname{sgn}(\lambda_N)$.
This proves \eqref{eq:select-action} on the $\mathsf B=N$ subspace.

By \eqref{eq:supported-reuse-lengths}, $\mathsf B\in\mathcal I_q$ if and only if
its value is even and lies in one of the intervals
\begin{equation*}
  [2, 2q - 2]\quad \text{or} \quad[2q + 2, 4q].
\end{equation*}
This test uses the least significant qubit of $\mathsf B$ and a constant number
of $b$-bit comparisons.
Computing $\mathsf B>\ell$ requires one further comparison.
Computing and erasing these conditions therefore uses $\bigO(b)$
gates~\cite[Section~4.3]{CuccaroEtAl04}.
Testing $\mathsf K=\ell$ reduces to a multi-controlled NOT and also uses
$\bigO(b)$ gates~\cite[Corollary~7.4]{BarencoEtAl95}.

Using the controlled implementation of $S^\circ$ from
\cref{prop:gate-efficient-update-product} in this SELECT construction gives the
gate count needed in the main theorem.
\begin{proposition}[Gate count for $\operatorname{SELECT}$]
\label{prop:gate-efficient-reuse-lcu}
Let $m=\log_2J$.
Each of $\operatorname{SELECT}$ and $\operatorname{SELECT}^\dagger$ uses $4q$
controlled $\HAMT$ queries and
\begin{equation}
  \bigO\! \bigl(q(a + m + \log(q + 1)) \bigr) = \bigO\! \bigl(q(a + \log J + 1)
      \bigr)
  \label{eq:gate-efficient-select-cost}
\end{equation}
one- and two-qubit gates.
Consequently, $V$ in \eqref{eq:review-lcu-circuit} has the same gate bound and
block-encodes $\widetilde U$ with normalization $2$, as in
\eqref{eq:review-lcu-block}.
\end{proposition}

\begin{proof}
Each of the $4q$ iterations contains one controlled $\HAMT$ query.
By \cref{prop:gate-efficient-update-product}, the controlled $S^\circ$ in one
iteration costs $\bigO(a+m)$ gates.
Computing $h_\ell$, testing $\mathsf K=\ell$, applying the controlled increment,
and erasing the work bits cost $\bigO(b)$ gates.
The $4q$ iterations and all controlled $F_N$ and $F_N^\dagger$ operations
therefore give
\begin{equation*}
  4q\, \bigO(a + m + b) + \bigO(qb) = \bigO\! \bigl(q(a + m + b) \bigr).
\end{equation*}
The phase test in item~1 costs $\bigO(b)$ gates and is absorbed in this bound.
Since $J\geq q$, we have $b=\bigO(m+1)$, proving
\eqref{eq:gate-efficient-select-cost}.
The $\operatorname{PREP}$ and $\operatorname{PREP}^\dagger$ circuits each cost
$\bigO(q)$ gates, so $V$ has the same asymptotic gate bound.
Every work qubit is returned to zero.
\end{proof}

\section{Proof of the main theorem}
\label{sec:proof}

The factorization \eqref{eq:dyadic-full-factorization} implements $S^\circ$, and
the circuits $V_N$ in \cref{sec:lcu} satisfy \eqref{eq:review-reuse-block} for
the operators $P_N$ analyzed in \cite{CGWZ26}.
Consequently, $V$ block-encodes $\widetilde U$ in
\eqref{eq:review-weighted-reuse}, so the error analysis of
\cite[Theorem~9]{CGWZ26} applies without change.
The choices of $q$ and $J$ in \cref{alg:review-cgwz} give error at most $\eps$,
where
\begin{equation}
  q = \bigO\! \left( \alpha T + \frac{\log(1/\eps)}
      {\log\! \left(e + \log(1/\eps)/(\alpha T) \right)} \right).
  \label{eq:main-q-order}
\end{equation}

We now match the query and gate counts to the steps of \cref{alg:review-cgwz}.
Set $m=\log_2J$ and let $b$ be the register width in
\eqref{eq:lcu-register-width}.
Steps~1--3 choose the parameters and define $U_C$.
In Steps~4--5, each call to $S$ uses one $\HAMT$ query and $\bigO(a+m)$
additional gates by \cref{prop:gate-efficient-update-product}.

Rather than implementing the circuits $V_N$ separately, Steps~5--7 are realized
together by
$V=\operatorname{PREP}^\dagger\operatorname{SELECT}\operatorname{PREP}$.
By \cref{prop:gate-efficient-reuse-lcu}, $V$ and $V^\dagger$ each use $4q$
$\HAMT$ queries and $\bigO(q(a+\log J+1))$ gates.
Since $J\geq q$, \eqref{eq:lcu-register-width} gives $b=\bigO(m+1)$.

Step~8 applies $V,V^\dagger,V$ and two reflections $\mathcal R_{\mathsf W}$.
Each reflection uses $\bigO(a+\log J+1)$ gates.
The resulting query count is $3(4q)=12q$, and the gate count is
\begin{equation}
  3 \bigO\! \bigl(q(a + \log J + 1) \bigr) + 2 \bigO(a + \log J + 1) = \bigO\!
      \bigl(q(a + \log J + 1) \bigr)
  \label{eq:full-gate-cost}
\end{equation}
one- and two-qubit gates.

It remains to express $\log J$ in terms of the problem parameters.
Recall that $J$ is the least power of two satisfying
\begin{equation}
  J \geq \max
      \left\{ 2, q, \frac{\beta T^2}{\eps},
      \frac{(\alpha T)^{3/2}}{\sqrt{3 \eps}} \right\}.
  \label{eq:main-J-choice}
\end{equation}
For $0<\eps\leq1/2$ and $\alpha T>\eps$, the denominator in the second term of
\eqref{eq:main-q-order} is at least $1$ when $\log(1/\eps)\leq\alpha T/\eps$,
and exceeds $\log(1/\eps)$ otherwise.
Hence
\begin{equation*}
  \frac{\log(1/\eps)} {\log\! \left(e + \log(1/\eps)/(\alpha T) \right)} \leq
      \frac{\alpha T}{\eps}.
\end{equation*}
Equation~\eqref{eq:main-q-order} therefore gives
\begin{equation*}
  q = \bigO\! \left(\frac{\alpha T}{\eps} \right).
\end{equation*}
The other two quantities in \eqref{eq:main-J-choice} satisfy
\begin{equation*}
  \frac{\beta T^2}{\eps} \leq 1 + \frac{T(\alpha + \beta T)}{\eps}, \qquad
      \frac{(\alpha T)^{3/2}}{\sqrt{3 \eps}} \leq \left(\frac{\alpha T}{\eps}
      \right)^{3/2} \leq \left(1 + \frac{T(\alpha + \beta T)}{\eps}
      \right)^{3/2}.
\end{equation*}
Because $J$ is the least power of two satisfying \eqref{eq:main-J-choice}, these
estimates imply
\begin{equation*}
  J = \bigO\! \left[ \left(1 + \frac{T(\alpha + \beta T)}{\eps} \right)^{3/2}
      \right],
\end{equation*}
and hence
\begin{equation*}
  \log J, \ \log(q + 1) = \bigO\! \left( \log\! \left(1 +
      \frac{T(\alpha + \beta T)}{\eps} \right) \right).
\end{equation*}
Substituting these logarithmic estimates and \eqref{eq:main-q-order} into
\eqref{eq:full-gate-cost} gives
\begin{equation*}
  \bigO\! \left[ \left( \alpha T + \frac{\log(1/\eps)}
      {\log\! \left(e + \log(1/\eps)/(\alpha T) \right)} \right) \left( a +
      \log\! \left(1 + \frac{T(\alpha + \beta T)}{\eps} \right) \right) \right],
\end{equation*}
which is the claimed gate bound.

Finally, the registers $\mathsf A,\mathsf T,\mathsf P,\mathsf K,\mathsf B$,
which remain allocated throughout the circuit, have widths $a,m,1,b,b$,
respectively.
At most $\bigO(a+m+b)$ work qubits are present at once; they are uncomputed
after use and reused.
The circuit therefore requires
\begin{align*}
  a + m + 1 + 2b + \bigO(a + m + b)
  &= \bigO(a + m + b)\\
  &= \bigO\! \left(
  a + \log\! \left(1 + \frac{T(\alpha + \beta T)}{\eps} \right)
  \right)
\end{align*}
auxiliary qubits beyond $\mathsf S$.
This completes the proof of \cref{thm:gate-efficient}.

\section{Time-integrated-norm scaling}
\label{sec:l1-scaling}

The gate-efficient implementation above is stated for the standard $\HAMT$ model
with fixed normalization, where the query complexity depends on the worst-case
bound $\alpha T$.
This can be pessimistic when the Hamiltonian is large only on a small part of
the time interval.
Following the time rescaling principle used to obtain $L^1$-norm scaling for
time-dependent Hamiltonian simulation~\cite[Section~4.1]{BCS+20}, we give a
conditional extension of our gate bound.
We assume direct coherent access to the reparameterized Hamiltonian defined
using a Lipschitz upper bound on $\norm{H(t)}$, which is stronger than the
standard $\HAMT$ oracle.
The resulting evolution length is proportional to the integral of this upper
bound.

\subsection{Time reparameterization and oracle access}
\label{sec:l1-access}

Let $T>0$, let $H\colon[0,T]\to\operatorname{Herm}(\Sys)$ be $\beta$-Lipschitz,
and let $h\colon[0,T]\to[0,\infty)$ be $\beta_h$-Lipschitz and satisfy
$\norm{H(t)}\leq h(t)$.
Define
\begin{equation}
  \Lambda_h := \int_0^T h(t) \dd{t}.
  \label{eq:l1-integrated-norm}
\end{equation}
Write $\norm{H}_{\infty,1}:=\int_0^T\norm{H(t)}\dd t$.
Then $\Lambda_h\geq\norm{H}_{\infty,1}$, with equality when $h(t)=\norm{H(t)}$.
If $\Lambda_h=0$, no reparameterization is needed: the identity circuit is
exact, as follows from \eqref{eq:l1-small-norm-bound} below.
Suppose henceforth that $\Lambda_h>0$, and set
\begin{equation}
  \mu := \frac{\Lambda_h}{T}, \qquad g(t) := h(t) + \mu, \qquad F(t) := \int_0^t
      g(u) \dd{u}.
  \label{eq:l1-clock}
\end{equation}
The addition of $\mu$ ensures that $g(t)>0$, even when $h(t)=0$.
Since $g(t)\geq\mu>0$ and $F(T)=2\Lambda_h$, the map $F$ is strictly increasing.
Define
\begin{equation}
  \tau := F^{-1} \colon [0, 2 \Lambda_h] \to [0, T], \qquad \widetilde H(s) :=
      \frac{H(\tau(s))}{g(\tau(s))}.
  \label{eq:l1-reparameterized-hamiltonian}
\end{equation}

\begin{definition}[Reparameterized Hamiltonian oracle]
For a power of two $J$, let $s_j=2j\Lambda_h/J$.
The reparameterized Hamiltonian oracle at resolution $J$ is the Hermitian
unitary
\begin{equation}
  \HAMT^{L^1}_J := \sum_{j = 0}^{J - 1} \ketbra{j}{j}_{\mathsf T} \otimes O_j,
  \label{eq:l1-hamt}
\end{equation}
where every $O_j$ acts on an $a$-qubit register $\mathsf A$ and the system, and
is a Hermitian unitary satisfying
\begin{equation}
  (\bra0_{\mathsf A} \otimes I_{\mathsf S})O_j (\ket0_{\mathsf A} \otimes
      I_{\mathsf S}) = \widetilde H(s_j).
  \label{eq:l1-hamt-block}
\end{equation}
Thus, on the uniform $s$-grid, the oracle block-encodes
$\widetilde H(s_j)=H(\tau(s_j))/g(\tau(s_j))$ with normalization one.
This is a stronger access model than the $\HAMT$ model with fixed normalization
in \cref{sec:preliminaries}, since it gives direct coherent access to the
normalized samples $H(\tau(s_j))/g(\tau(s_j))$.
\end{definition}

\begin{lemma}[Time reparameterization]
\label{lem:time-reparameterization}
Let $T>0$, let $H\colon[0,T]\to\operatorname{Herm}(\Sys)$ be $\beta$-Lipschitz,
and let $h\colon[0,T]\to[0,\infty)$ be $\beta_h$-Lipschitz and satisfy
$\norm{H(t)}\leq h(t)$.
Assume that $\Lambda_h$ in \eqref{eq:l1-integrated-norm} is positive, and define
$\mu,g,F,\tau$, and $\widetilde H$ by
\eqref{eq:l1-clock}--\eqref{eq:l1-reparameterized-hamiltonian}.
Then:
\begin{enumerate}
 \item $F$ is a strictly increasing $C^1$ bijection from $[0,T]$
 onto $[0,2\Lambda_h]$, and
 \begin{equation*}
   \tau'(s) = \frac{1}{g(\tau(s))}.
 \end{equation*}
 \item The reparameterized Hamiltonian satisfies
 \begin{equation*}
   \norm{\widetilde H(s)} \leq 1 \qquad(0 \leq s \leq 2 \Lambda_h).
 \end{equation*}
 \item It is Lipschitz with a constant $\widetilde\beta$ satisfying
 \begin{equation*}
   \norm{\widetilde H(s) - \widetilde H(r)} \leq \widetilde \beta|s - r|, \qquad
       \widetilde \beta \leq \frac{(\beta + \beta_h)T^2}{\Lambda_h^2}.
 \end{equation*}
 \item The full time-ordered evolution is preserved:
 \begin{equation*}
   U_{\widetilde H}(2 \Lambda_h) = U_H(T).
 \end{equation*}
\end{enumerate}
\end{lemma}

\begin{proof}
The function $g$ is continuous, so $F$ is $C^1$, with $F'(t)=g(t)$.
If $0\leq u<t\leq T$, then
\begin{align*}
  F(t) - F(u)
  &= \int_u^t g(v) \dd{v}\\
  & \geq \mu(t - u) > 0.
\end{align*}
Thus $F$ is strictly increasing.
Moreover,
\begin{equation*}
  F(0) = 0, \qquad F(T) = \int_0^T h(t) \dd{t} + \mu T = 2 \Lambda_h.
\end{equation*}
Continuity and strict monotonicity therefore make $F$ a bijection from $[0,T]$
onto $[0,2\Lambda_h]$.
Since $F'=g\geq\mu>0$, the inverse function theorem shows that its inverse is
$C^1$ and
\begin{align*}
  \tau'(s)
  &= \frac{1}{F'(\tau(s))}\\
  &= \frac{1}{g(\tau(s))}.
\end{align*}
For every $s\in[0,2\Lambda_h]$,
\begin{equation*}
  \norm{\widetilde H(s)} = \frac{\norm{H(\tau(s))}}{h(\tau(s)) + \mu} \leq 1.
\end{equation*}
For $t,u\in[0,T]$,
\begin{equation*}
  |g(t) - g(u)| = |h(t) - h(u)| \leq \beta_h|t - u|.
\end{equation*}
Consequently,
\begin{align*}
  \norm{\frac{H(t)}{g(t)} - \frac{H(u)}{g(u)}}
  & \leq
  \frac{\norm{H(t) - H(u)}}{g(t)}
  + \frac{\norm{H(u)}\, |g(t) - g(u)|}{g(t)g(u)} \\
  & \leq
  \frac{\beta}{\mu}|t - u|
  + \frac{\norm{H(u)}}{g(u)} \frac{\beta_h}{\mu}|t - u| \\
  & \leq \frac{\beta + \beta_h}{\mu}|t - u|.
\end{align*}
Also,
\begin{equation*}
  |F(t) - F(u)| \geq \mu|t - u|.
\end{equation*}
Taking $t=\tau(s)$ and $u=\tau(r)$ yields
\begin{equation*}
  |\tau(s) - \tau(r)| \leq \frac{|s - r|}{\mu}.
\end{equation*}
Combining the last two estimates gives
\begin{equation*}
  \norm{\widetilde H(s) - \widetilde H(r)} \leq \frac{\beta + \beta_h}{\mu^2}|s
      - r| = \frac{(\beta + \beta_h)T^2}{\Lambda_h^2}|s - r|,
\end{equation*}
which proves the claimed Lipschitz bound.

Finally, define $V(s):=U_H(\tau(s))$.
The chain rule and the defining equation for $U_H$ give
\begin{align*}
  i \dv{}{s}V(s)
  &= H(\tau(s)) \tau'(s)V(s)\\
  &= \widetilde H(s)V(s),
  \qquad
  V(0) = I.
\end{align*}
By uniqueness of the propagator, $V(s)=U_{\widetilde H}(s)$.
Since $\tau(2\Lambda_h)=T$, this implies
\begin{equation*}
  U_{\widetilde H}(2 \Lambda_h) = U_H(T).
\end{equation*}
\end{proof}

\paragraph{The small-\texorpdfstring{$\Lambda_h$}{Lambda h} regime.}
The integral equation for the propagator is
\begin{equation}
  U_H(T) - I = -i \int_0^T H(t)U_H(t) \dd{t}.
  \label{eq:l1-propagator-integral}
\end{equation}
Since $U_H(t)$ is unitary,
\begin{equation}
  \norm{U_H(T) - I} \leq \int_0^T \norm{H(t)} \dd{t} \leq \Lambda_h.
  \label{eq:l1-small-norm-bound}
\end{equation}
Thus, if $T=0$ or $\Lambda_h\leq\eps$, the identity circuit gives error at most
$\eps$ and makes no oracle queries.

\subsection{Gate-efficient integrated-norm scaling}
\label{sec:l1-corollary}

\begin{theorem}[Gate-efficient integrated-norm scaling]
\label{thm:l1-gate-efficient}
Let $T>0$, let $H\colon[0,T]\to\operatorname{Herm}(\Sys)$ be $\beta$-Lipschitz,
and let $h\colon[0,T]\to[0,\infty)$ be $\beta_h$-Lipschitz and satisfy
$\norm{H(t)}\leq h(t)$.
Let $0<\eps\leq1/2$, set
\begin{equation*}
  \Lambda_h := \int_0^T h(t) \dd{t},
\end{equation*}
and suppose that $\Lambda_h>\eps$.
Given query access to $\HAMT^{L^1}_J$ in \eqref{eq:l1-hamt} for a sufficiently
large power of two $J$, with an $a$-qubit block-encoding register, there is a
circuit $U_{\mathrm{sim}}$ and an auxiliary register $\mathsf W$ such that, for
every state $\ket\psi\in\Sys$,
\begin{equation*}
  \norm{ U_{\mathrm{sim}}(\ket0_{\mathsf W} \ket \psi) -
      \ket0_{\mathsf W}U_H(T) \ket \psi}
      \leq \eps.
\end{equation*}
It uses $\bigO(q_{L^1})$ queries, where
\begin{equation}
  q_{L^1} := \Lambda_h + \frac{\log(1/\eps)}
      {\log\! \left(e + \log(1/\eps)/\Lambda_h \right)}.
  \label{eq:l1-query-bound}
\end{equation}
In addition, the circuit uses
\begin{equation*}
  n_{\mathrm{anc}} = \bigO\! \left( a + \log\! \left( 1 +
      \frac{\Lambda_h + (\beta + \beta_h)T^2}{\eps} \right) \right)
\end{equation*}
auxiliary qubits and $\bigO(q_{L^1}n_{\mathrm{anc}})$ additional one- and
two-qubit gates.
\end{theorem}

The improvement in \cref{thm:l1-gate-efficient} is the replacement of the
worst-case evolution length $\alpha T$ by the upper bound $\Lambda_h$ on the
integrated norm.
In particular, $h(t)=\norm{H(t)}$ gives $\Lambda_h=\norm{H}_{\infty,1}$ and
satisfies $\beta_h\leq\beta$.
The Lipschitz constants remain in the logarithmic gate and space factor through
the resolution of the reparameterized Hamiltonian.

\begin{proof}
Apply \cref{thm:gate-efficient} to the reparameterized Hamiltonian
$\widetilde H$ on $[0,2\Lambda_h]$, with normalization $1$ and Lipschitz
constant $\widetilde\beta\leq(\beta+\beta_h)T^2/\Lambda_h^2$.
\Cref{lem:time-reparameterization} gives $\norm{\widetilde H(s)}\leq1$, the
stated Lipschitz bound, and $U_{\widetilde H}(2\Lambda_h)=U_H(T)$.
The oracle $\HAMT^{L^1}_J$ is the $\HAMT$ oracle required by
\cref{thm:gate-efficient} for $\widetilde H$ on the uniform grid
$s_j=2j\Lambda_h/J$, with normalization one.

For the query complexity, the normalization is one and the evolution time is
$2\Lambda_h$, so \cref{thm:gate-efficient} gives
\begin{equation*}
  q = \bigO\! \left( 2 \Lambda_h + \frac{\log(1/\eps)}
      {\log\! \left(e + \log(1/\eps)/(2 \Lambda_h) \right)} \right).
\end{equation*}
For every $z\geq0$,
\begin{equation*}
  e + z \leq (e + z/2)^2, \qquad \text{and hence} \qquad \log(e + z/2) \geq
      \frac12 \log(e + z).
\end{equation*}
Taking $z=\log(1/\eps)/\Lambda_h$ absorbs the factor two in the denominator,
while $2\Lambda_h=\bigO(\Lambda_h)$.
Thus $q=\bigO(q_{L^1})$, as claimed.

For the gate complexity, the logarithmic factor in \cref{thm:gate-efficient}
contains
\begin{align*}
  (2 \Lambda_h) \bigl(1 + 2 \Lambda_h \widetilde \beta \bigr) & \leq 2 \Lambda_h
      + 4(\beta + \beta_h)T^2.
\end{align*}
Consequently,
\begin{equation*}
  \log\! \left( 1 + \frac{(2 \Lambda_h)(1 + 2 \Lambda_h \widetilde \beta)}{\eps}
      \right) = \bigO\! \left( \log\! \left( 1 +
      \frac{\Lambda_h + (\beta + \beta_h)T^2}{\eps} \right) \right).
\end{equation*}
Constant factors inside the logarithm are absorbed because $\Lambda_h/\eps>1$.
Substitution in \cref{thm:gate-efficient}, together with $q=\bigO(q_{L^1})$,
gives the asserted gate bound.
The same substitution in its auxiliary-qubit bound gives the asserted space
bound.

A sufficient resolution is obtained by letting $q$ denote the integer selected
by \cref{alg:review-cgwz} after this substitution and taking $J$ to be a power
of two.
The resolution condition used in \cref{thm:gate-efficient} is satisfied when
\begin{equation*}
  J \geq \max
      \left\{ 2, q, \frac{4(\beta + \beta_h)T^2}{\eps},
      \frac{(2 \Lambda_h)^{3/2}}{\sqrt{3 \eps}} \right\}.
\end{equation*}
Indeed, the third term bounds $4\Lambda_h^2\widetilde\beta/\eps$, and the fourth
is exactly $(2\Lambda_h)^{3/2}/\sqrt{3\eps}$.
Finally, $U_{\widetilde H}(2\Lambda_h)=U_H(T)$, so the simulator produced by
\cref{thm:gate-efficient} has the stated target.
\end{proof}

\section*{Acknowledgments}

Boyang Chen, Minbo Gao and Zhengfeng Ji are supported by National Key Research
and Development Program of China (Grant No.~2023YFA1009403) and National
Natural Science Foundation of China (Grant No.~12347104).
Tongyang Li, Xinzhao Wang, and Shuo Zhou are supported by the National Natural
Science Foundation of China (Grant No.~62372006).

\paragraph{Statement on AI use.}
The project was initiated and motivated by the authors.
The main ideas underlying the work, including the central proof strategies, were
generated by large language models.
The human authors carefully studied, refined, and verified these ideas, and
substantially rewrote and reorganized their presentation, adding the motivation
and exposition necessary to communicate the arguments clearly.
The final paper reflects the authors' own understanding of the results, and the
authors take full responsibility for every claim, proof, and citation in it.

\begingroup
\small
\bibliographystyle{alphaurl}
\bibliography{ref}
\endgroup

\end{document}